\documentclass[12pt,a4paper,oneside]{article}
\usepackage[affil-it]{authblk}
\usepackage[centertags]{amsmath}
\usepackage{amssymb,mathtools,mathrsfs}
\usepackage{amsthm}
\usepackage{enumerate}
\usepackage{stmaryrd}
\usepackage{dsfont}
\usepackage{url}
\usepackage{graphicx}
\usepackage{color}
\usepackage{epsfig}
\usepackage{subcaption}
\usepackage{geometry}
\definecolor{dg}{rgb}{0,0.5,0}		
\definecolor{orange}{rgb}{1.0, 0.5, 0.0}

\usepackage{hyphenat}
\usepackage[english]{babel}
\usepackage[T1]{fontenc}

\theoremstyle{theorem}
\newtheorem{Thm}{Theorem}

\newtheorem{Prop}[Thm]{Proposition}

\theoremstyle{definition}

\newtheorem{Rem}[Thm]{Remark}

\newcommand{\ii}{\textup{i}}
\newcommand{\ee}{\textup{e}}

\newcommand{\jdn}{\textbf{1}}

\DeclareMathOperator{\diag}{diag}

\DeclareMathOperator{\Tr}{Tr}

\DeclareMathOperator{\RE}{Re}

\renewcommand{\bar}{\overline}		

\begin{document}


\title{Comment on `Quantum Monge--Kantorovich Problem and Transport Distance between Density Matrices'}
\author{Tomasz Miller${}^{1}$\thanks{tomasz.miller@uj.edu.pl}}

\affil{\small ${}^1$ Copernicus Center for Interdisciplinary Studies, Jagiellonian University,
\\Szczepa\'nska 1/5, 31-011 Krak\'ow, Poland}
\date{\today}


\maketitle

\begin{abstract}
Friedland et al. [\emph{Phys. Rev. Lett.} \textbf{129}, 110402 (2022)] proposed and studied a quantum analogue of the $p$-Wasserstein distance based on quantum cost matrices and quantum couplings. They conjectured that, despite being only a semidistance in general, this quantity is a true distance for a particular quantum cost matrix and for cost matrices in a small neighborhood of it. We disprove these conjectures by exhibiting an explicit family of triples of states for which the triangle inequality fails.
\end{abstract}



\section{Introduction}

In \cite{QOT} the following construction was put forward. For any $n \times n$ distance matrix $(E_{ij})$ define its associated \emph{quantum cost matrix} $C_E$ given by (cf. \cite[(15)]{QOT})
\begin{align}
\label{costmatrix}
C_E := \sum_{i<j}E_{ij}|i \wedge j\rangle \langle i \wedge j|,
\end{align}
where $| i \wedge j \rangle := \tfrac{1}{\sqrt{2}} (|ij\rangle - |ji\rangle)$ with $\{|i\rangle\}_{i=1}^n$ being some fixed computational basis, and then for any chosen $p \geq 1$ consider the quantum analogue of the $p$-Wasserstein distance (cf. \cite[(16)]{QOT})
\begin{align}
\label{Wdef}
W_{C_E,p}(\rho_1,\rho_2) := \min_{\rho_{12} \in \Gamma(\rho_1,\rho_2)} \left[ \Tr C^p_E\rho_{12} \right]^{1/p} = \min_{\rho_{12} \in \Gamma(\rho_1,\rho_2)} \left[ \sum_{i<j} E^p_{ij}\langle i \wedge j | \rho_{12} | i \wedge j \rangle \right]^{1/p}
\end{align}
for any $n \times n$ density matrices $\rho_1,\rho_2$, where the minimum runs over the set $\Gamma(\rho_1,\rho_2)$ of all their quantum couplings, i.e., $n^2 \times n^2$ density matrices $\rho_{12}$ whose partial traces are $\Tr_1 \rho_{12} = \rho_2$ and $\Tr_2 \rho_{12} = \rho_1$. It has been shown \cite[Proposition 4]{QOT} that $W_{C_E,p}$ is a semidistance, but in general it violates the triangle inequality \cite[Appendix L]{QOT}. On the other hand, there are some intriguing special cases where $W_{C_E,p}$ \emph{does} satisfy the triangle inequality and thus is a true distance. These include the $n=2$ case \cite[Theorem 3]{QOT} and the case where $W_{C_E,p}$ is restricted to pure states \cite{Triangle}, in both cases under the necessary assumption that $p \geq 2$.

Based on numerical investigations, it has been conjectured \cite[Conjecture I]{QOT} that if the quantum cost matrix is chosen to be the projector onto the antisymmetric subspace $C^Q := \sum_{i<j}|i \wedge j\rangle \langle i \wedge j| = \tfrac{1}{2}(\jdn_{n^2}-S)$, where $S$ is the SWAP operator, then the associated quantum $2$-Wasserstein semidistance
\begin{align}
\label{Wdef2}
W_{C^Q,2}(\rho_1,\rho_2) = \min_{\rho_{12} \in \Gamma(\rho_1,\rho_2)} \sqrt{ \sum_{i<j} \langle i \wedge j | \rho_{12} | i \wedge j \rangle }
\end{align}
satisfies the triangle inequality for \emph{any} dimension $n$. What is more, subsequent Conjecture II claimed that this property remains true when the quantum cost matrix $C^Q$ is slightly perturbed. In this note, however, we disprove both conjectures by presenting analytical counterexamples---a broad family of triples of states (\ref{ex}) for which the triangle inequality is violated.

\section{Results}

The counterexamples involve certain \emph{diagonal} density matrices, and in order to study them we shall need the following general formula expressing $W_{C_E,p}$ between such classical states in terms of a minimum over classical couplings\footnote{It generalizes the result obtained in the final part of the proof of \cite[Theorem 4.3]{QOT2} and it also sheds some light on \cite[Theorem C.1]{QOT2}.}.
\begin{Prop}
Let $P = (p_1,\ldots,p_n)$ and $Q = (q_1,\ldots,q_n)$ be probability vectors and let $(E_{ij})$ be a real symmetric $n \times n$ matrix. Then
\begin{align}
\label{wzor}
W_{C_E,p}(\diag(P),\diag(Q)) = \min_{\gamma \in \Gamma_{\textup{cl}}(P,Q)} \left[ \tfrac{1}{2} \sum_{i<j}E_{ij}^p (\sqrt{\gamma_{ij}}-\sqrt{\gamma_{ji}})^2 \right]^{1/p},
\end{align}
where the minimum runs over all \emph{classical} couplings of $P$ and $Q$.
\end{Prop}
\begin{proof}
\textbf{Step 1.} For any $\theta := (\theta_1,\ldots,\theta_n) \in [0,2\pi)^n$ define the diagonal unitary operator $U_\theta := \sum_{k=1}^n \ee^{\ii \theta_k} |k\rangle \langle k|$. Denote $\rho := \diag(P)$, $\sigma := \diag(Q)$ and for any $\omega \in \Gamma(\rho,\sigma)$ consider $\omega_\theta := (U_\theta^+ \otimes U_\theta^+ )\omega (U_\theta \otimes U_\theta)$ as well as its average over all $\theta$
\begin{align}
\widetilde{\omega} := \frac{1}{(2\pi)^n}\int_{[0,2\pi)^n} \omega_\theta d\theta.
\end{align}
Observe that $\omega_\theta \in \Gamma(\rho,\sigma)$ for any $\theta$, and hence also $\widetilde{\omega} \in \Gamma(\rho,\sigma)$. Indeed, one has that
\begin{align}
\langle ij | \omega_\theta  | kl \rangle = \ee^{\ii (\theta_k + \theta_l - \theta_i - \theta_j)} \langle ij | \omega | kl \rangle,
\end{align}
and hence the partial traces of $\omega_\theta$ have the matrix elements
\begin{align}
& \langle i | \Tr_2 \omega_\theta | k \rangle = \sum_j \langle ij | \omega_\theta | kj \rangle = \ee^{\ii (\theta_k - \theta_i)} \sum_j \langle ij | \omega | kj \rangle = \ee^{\ii (\theta_k - \theta_i)} \langle i | \rho | k \rangle = \langle i | \rho | k \rangle
\end{align}
and, similarly,
\begin{align}
& \langle j | \Tr_1 \omega_\theta | l \rangle = \sum_i \langle ij | \omega_\theta | il \rangle = \ee^{\ii (\theta_l - \theta_j)} \sum_i \langle ij | \omega | il \rangle = \ee^{\ii (\theta_l - \theta_j)} \langle j | \sigma | l \rangle = \langle j | \sigma | l \rangle,
\end{align}
where we have used the diagonality of both $\rho$ and $\sigma$.

What is more, the transport cost does not change due to $U_\theta$-averaging,
\begin{align}
\sum_{i<j} E^p_{ij}\langle i \wedge j | \widetilde{\omega} | i \wedge j \rangle = \sum_{i<j} E^p_{ij}\langle i \wedge j | \omega | i \wedge j \rangle
\end{align}
simply because $(U_\theta \otimes U_\theta)| i \wedge j \rangle = \ee^{\ii (\theta_i + \theta_j)} | i \wedge j \rangle$ for any $i,j$ and any $\theta$.

In light of the above, when calculating $W_{C_E,p}$ we can minimize only over $U_\theta$-averaged couplings $\widetilde{\omega}$.

\textbf{Step 2.} We claim that every $U_\theta$-averaged coupling $\widetilde{\omega}$ has a block-diagonal matrix in the computational basis. Concretely, notice that
\begin{align}
\langle ij | \widetilde{\omega} | kl \rangle = \frac{\langle ij | \omega | kl \rangle}{(2\pi)^n}\int_{[0,2\pi)^n} \ee^{\ii (\theta_k + \theta_l - \theta_i - \theta_j)} d\theta,
\end{align}
which can be nonzero only if $\{i,j\} = \{k,l\}$ as (multi)sets. In other words, the only potentially nonzero matrix elements of $\widetilde{\omega}$ are
\begin{align}
\nonumber
& \langle ii | \widetilde{\omega} | ii \rangle =: \gamma_{ii}, \quad \langle ij | \widetilde{\omega} | ij \rangle =: \gamma_{ij}, \quad \langle ji | \widetilde{\omega} | ji \rangle =: \gamma_{ji},
\\
\label{step2}
& \langle ij | \widetilde{\omega} | ji \rangle =: z_{ij} \quad \textnormal{and} \quad  \langle ji | \widetilde{\omega} | ij \rangle = \bar{z}_{ij},
\end{align}
for any $1 \leq i < j \leq n$, where in the last line we have used the hermiticity of $\widetilde{\omega}$, which also implies that $(\gamma_{ij})$ is a real matrix. What is more, since $\widetilde{\omega}$ is positive semi-definite, we have $\gamma_{ii},\gamma_{ij},\gamma_{ji} \geq 0$ and $|z_{ij}|^2 \leq \gamma_{ij}\gamma_{ji}$ for all $i,j$ such that $i < j$.

Moreover, $(\gamma_{ij})$ is actually a classical coupling of $P$ and $Q$ in the sense that
\begin{align*}
\sum_j \gamma_{ij} = \sum_j \langle ij | \widetilde{\omega} | ij \rangle = \langle i | \rho | i \rangle = p_i \quad \textnormal{and} \quad \sum_i \gamma_{ij} = \sum_i \langle ij | \widetilde{\omega} | ij \rangle = \langle j | \sigma | j \rangle = q_j. 
\end{align*}

\textbf{Step 3.} For any $U_\theta$-averaged coupling $\widetilde{\omega}$ the transport cost can be now easily calculated and then bounded from below as
\begin{align}
\label{step3}
& \sum_{i<j} E^p_{ij}\langle i \wedge j | \widetilde{\omega} | i \wedge j \rangle = \tfrac{1}{2}\sum_{i<j} E^p_{ij} \left( \gamma_{ij} + \gamma_{ji} - 2\RE z_{ij} \right)
\\
\nonumber
& \geq \tfrac{1}{2}\sum_{i<j} E^p_{ij} \left( \gamma_{ij} + \gamma_{ji} - 2\sqrt{\gamma_{ij}\gamma_{ji}} \right) = \tfrac{1}{2}\sum_{i<j} E^p_{ij} \left( \sqrt{\gamma_{ij}} - \sqrt{\gamma_{ji}} \right)^2
\end{align}

Bearing in mind all of the above, we obtain that 
\begin{align*}
W_{C_E,p}(\rho,\sigma) & = \min_{\widetilde{\omega}} \left[ \sum_{i<j} E^p_{ij}\langle i \wedge j | \widetilde{\omega} | i \wedge j \rangle \right]^{1/p} \geq \min_{\gamma \in \Gamma_{\textup{cl}}(P,Q)} \left[ \tfrac{1}{2} \sum_{i<j}E_{ij}^p (\sqrt{\gamma_{ij}}-\sqrt{\gamma_{ji}})^2 \right]^{1/p},
\end{align*}
but the rightmost minimum is always attained --- indeed, for any $\gamma \in \Gamma_{\textup{cl}}(P,Q)$ simply \emph{define} $\widetilde{\omega}$ via (\ref{step2}) with $z_{ij} := \sqrt{\gamma_{ij}\gamma_{ji}}$ for any $i,j$ such that $i<j$.
\end{proof}

\begin{Rem}
\cite[Proposition 5]{QOT} states that the quantum transport between classical (i.e., diagonal) states is always cheaper than its classical counterpart. The above result makes this statement more precise. Concretely, observe that
\begin{align*}
W_{C_E,p}(\diag(P),\diag(Q)) & = \min_{\gamma \in \Gamma_{\textup{cl}}(P,Q)} \left[ \tfrac{1}{2} \sum_{i<j}E_{ij}^p (\sqrt{\gamma_{ij}}-\sqrt{\gamma_{ji}})^2 \right]^{1/p}
\\
& = \min_{\gamma \in \Gamma_{\textup{cl}}(P,Q)} \left[ \tfrac{1}{2} \sum_{i<j}E_{ij}^p \left( \gamma_{ij} + \gamma_{ji} - 2\sqrt{\gamma_{ij}\gamma_{ji}} \right) \right]^{1/p}
\\
& \leq \min_{\gamma \in \Gamma_{\textup{cl}}(P,Q)} \left[ \tfrac{1}{2} \sum_{i<j}E_{ij}^p \left( \gamma_{ij} + \gamma_{ji} \right) \right]^{1/p} = \min_{\gamma \in \Gamma_{\textup{cl}}(P,Q)} \left[ \tfrac{1}{2} \sum_{ij} E_{ij}^p \gamma_{ij} \right]^{1/p},
\end{align*}
where the rightmost expression is the classical $p$-Wasserstein distance\footnote{Provided $E$ is a distance matrix.} up to the $2^{-1/p}$ factor, which perhaps should have been taken into account already in \cite{QOT} in the very definition of $W_{C_E,p}$. Concretely, definition (\ref{Wdef}) should rather involve $\min_{\rho_{12} \in \Gamma(\rho_1,\rho_2)} \left[ 2 \Tr C^p_E\rho_{12} \right]^{1/p}$, so that when considering classical states and restricting the minimization to the classical couplings $\omega_\textup{cl} = \sum_{ij}\gamma_{ij}|ij\rangle \langle ij|$, for cost matrices (\ref{costmatrix}) one would indeed recover the classical $p$-Wasserstein distance.

Bearing that in mind, we can thus clearly see \emph{why} the quantum transport is cheaper --- it is precisely due to the off-diagonal ``interference'' terms $\sqrt{\gamma_{ij}\gamma_{ji}}$.
\end{Rem}

\begin{Rem}
Formula (\ref{wzor}) for $n=2$, $p=2$ and $E=\left[ \begin{smallmatrix} 0 & 1 \\ 1 & 0 \end{smallmatrix} \right]$ offers a way of rederiving \cite[(11)]{QOT} (or, conversely, the latter can serve as a sanity check of (\ref{wzor}) in the $n=2$ case).

Indeed, let us take $\rho = \diag(r,1-r)$ and $\sigma = \diag(s,1-s)$. Then their classical couplings form the one-parameter family
\begin{align*}
\gamma = \begin{bmatrix} \lambda & p-\lambda \\ q-\lambda & 1-p-q+\lambda \end{bmatrix}, \qquad \max\{0, p+q-1\} \leq \lambda \leq \min\{p,q\}
\end{align*}
and so (\ref{wzor}) yields
\begin{align}
\label{sanity}
W_{C_E,2}(\rho,\sigma) = \tfrac{1}{\sqrt{2}}  \min_{\lambda} \left| \sqrt{p-\lambda}-\sqrt{q-\lambda} \right| = \left\{ \begin{array}{ll} \tfrac{1}{\sqrt{2}} \left| \sqrt{p}-\sqrt{q} \right| & \textnormal{if } p+q \leq 1 
\\
\tfrac{1}{\sqrt{2}} \left| \sqrt{1-p}-\sqrt{1-q} \right| & \textnormal{if } p+q > 1 \end{array} \right. ,
\end{align}
where we have used the fact that the map $\lambda \mapsto \left| \sqrt{p-\lambda}-\sqrt{q-\lambda} \right|$ is increasing on the considered interval. But the rightmost expression in (\ref{sanity}) can be shown to be equal to $\tfrac{1}{\sqrt{2}} \max\{ \left| \sqrt{p}-\sqrt{q} \right|, \left| \sqrt{1-p}-\sqrt{1-q} \right|\}$ (i.e., to \cite[(11)]{QOT}), by proving the equivalence
\begin{align*}
p+q \leq 1 \quad \Longleftrightarrow \quad \left| \sqrt{p}-\sqrt{q} \right| \geq \left| \sqrt{1-p}-\sqrt{1-q} \right|
\end{align*}
for any $p,q \in [0,1]$, what can be done by tedious but straightforward calculations.
\end{Rem}

With formula (\ref{wzor}) at hand, we are finally ready to present a large family of triples of diagonal states for which $W_{C_E,2}$ violates the triangle inequality. Concretely, let
\begin{align}
\label{ex}
& \rho := \diag(1-s,s,0) \qquad \sigma := \diag(s,1-s-t,t) \qquad \tau := \diag(t,1-s-t,s),
\\
\nonumber
& \textnormal{where the parameter space is} \quad \Delta := \{ (s,t) \ | \ 0 \leq t \leq s \leq \tfrac{1}{2}, \ 2s+t \leq 1 \}.
\end{align}
Denoting $W_{C_E,2}$ simply by $W$, we shall prove first that
\begin{align}
\label{dist1}
& W(\rho,\sigma) = \sqrt{\tfrac{1}{2} -\sqrt{s(1-s-t)}},
\\
\label{dist2}
& W(\rho,\tau) = \sqrt{\tfrac{1}{2}(1+t-s) - \sqrt{t(1-2s)}},
\\
\label{dist3}
& W(\sigma,\tau) = \tfrac{1}{\sqrt{2}}\left(\sqrt{s}-\sqrt{t}\right)
\end{align}
and then demonstrate that $W(\rho,\tau) > W(\rho,\sigma) + W(\sigma,\tau)$ on the \emph{entire interior} of $\Delta$ (see Fig. \ref{fig1} for illustration).

\begin{figure}[h]
\label{fig1}
\begin{center}
\resizebox{350pt}{!}{\includegraphics[scale=1.5]{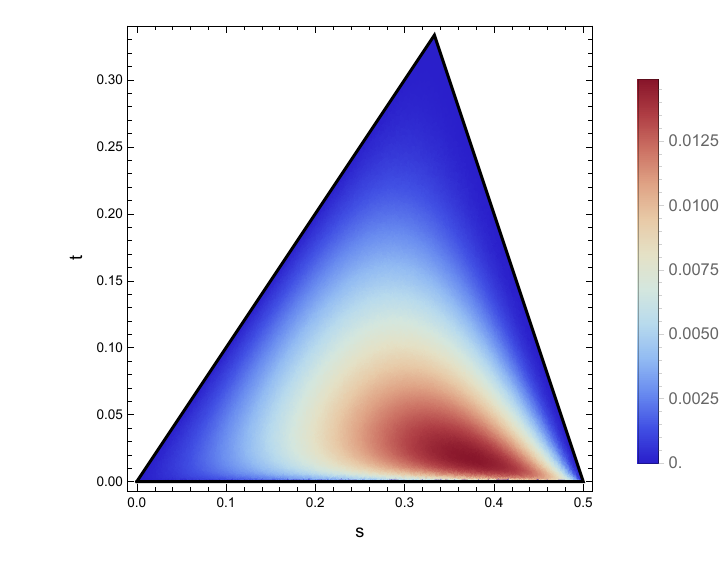}}
\caption{\label{fig1}Graph of $W(\rho,\tau) - W(\rho,\sigma) - W(\sigma,\tau)$ given by formulas (\ref{dist1}--\ref{dist3}) on the admissible parameter space $\Delta = \{ (s,t) \ | \ 0 \leq t \leq s \leq \tfrac{1}{2}, \ 2s+t \leq 1 \}$. The map vanishes on $\partial \Delta$ and is positive on $\Delta^\circ$.}
\end{center}
\end{figure}

Although the above states are 3-dimensional, they of course provide counterexamples for any $n > 3$ as well --- simply consider $\rho \oplus \textbf{0}_{n-3}$, $\sigma \oplus \textbf{0}_{n-3}$ and $\tau \oplus \textbf{0}_{n-3}$.

In order to prove (\ref{dist1}), notice first that the classical couplings of $P := (1-s,s,0)$ and $Q := (s,1-s-t,t)$ can be written as
\begin{align*}
\gamma = \begin{bmatrix} \lambda & 1-s-t-\lambda+\mu & t-\mu \\ s-\lambda & \lambda-\mu & \mu \\ 0 & 0 & 0 \end{bmatrix},
\end{align*}
where $\lambda,\mu$ are parameters. The conditions $\gamma_{ij} \in [0,1]$ can be shown to restrict the parameter space to $\lambda \in [0,s]$, $\mu \in [0,\min\{\lambda,t\}]$. Formula (\ref{wzor}) now gives
\begin{align*}
W^2(\rho,\sigma) & = \tfrac{1}{2} \min_{\lambda,\mu} \left[\left( \sqrt{1-s-t-\lambda+\mu} - \sqrt{s-\lambda} \right)^2 + t \right]
\\
& = \tfrac{1}{2} \min_{\lambda} \left[\left( \sqrt{1-s-t-\lambda} - \sqrt{s-\lambda} \right)^2 + t \right] = \tfrac{1}{2} \left[\left( \sqrt{1-s-t} - \sqrt{s} \right)^2 + t \right],
\end{align*}
where we have first used the fact that the function under $\min_{\lambda,\mu}$ is nondecreasing in $\mu$ and so we could simply plug $\mu = 0$, and then the fact that the function under $\min_{\lambda}$ is nondecreasing in $\lambda$ and so attains its minimum at $\lambda = 0$ --- the lowest admissible parameter values. This yields (\ref{dist1}).

Proof of (\ref{dist2}) goes along similar lines. The classical couplings of $P := (1-s,s,0)$ and $R := (t,1-s-t,s)$ can be written as
\begin{align*}
\gamma = \begin{bmatrix} \lambda & 1-2s-\lambda+\mu & s-\mu \\ t-\lambda & s-t+\lambda-\mu & \mu \\ 0 & 0 & 0 \end{bmatrix},
\end{align*}
where the conditions $\gamma_{ij} \in [0,1]$ this time restrict the parameter space to $\lambda \in [0,t]$, $\mu \in [0,s-t+\lambda]$. Formula (\ref{wzor}) allows us to write
\begin{align*}
W^2(\rho,\tau) & = \tfrac{1}{2} \min_{\lambda,\mu} \left[\left( \sqrt{1-2s-\lambda+\mu} - \sqrt{t-\lambda} \right)^2 + s \right] = \tfrac{1}{2} \left[\left( \sqrt{1-2s} - \sqrt{t} \right)^2 + s \right],
\end{align*}
where again we have first noticed that the function under $\min$ is nondecreasing in $\mu$ and therefore plugged $\mu = 0$ and then done the same with respect to $\lambda$. This yields (\ref{dist2}).

Proving (\ref{dist3}) requires a bit more care. Instead of parametrizing the couplings $\gamma$ of $Q$ and $R$ (what would require as much as 4 parameters), let us first notice that the map $\Phi(\gamma) := \sum_{1 \leq i<j \leq 3} (\sqrt{\gamma_{ij}} - \sqrt{\gamma_{ji}} )^2$ can be bounded from below as
\begin{align*}
\Phi(\gamma) & \geq (\sqrt{\gamma_{12}} - \sqrt{\gamma_{21}} )^2 + (\sqrt{\gamma_{13}} - \sqrt{\gamma_{31}} )^2 = \sum_{j=1}^3 (\sqrt{\gamma_{1j}} - \sqrt{\gamma_{j1}} )^2 
\\
& \geq \left( \sqrt{\textstyle\sum_{j=1}^3 \gamma_{1j}} - \sqrt{\textstyle\sum_{j=1}^3 \gamma_{j1}}\right)^2 = \left(\sqrt{s} - \sqrt{t}\right)^2,
\end{align*} 
where in the first inequality we simply omitted a (nonnegative) summand, in the second inequality used Cauchy--Schwarz inequality, and in the last equality we substituted the appropriate values of the marginals. This means that, by (\ref{wzor}),
\begin{align*}
W^2(\sigma,\tau) = \tfrac{1}{2} \min_{\gamma \in \Gamma_{\textup{cl}}(Q,R)} \Phi(\gamma) \geq \tfrac{1}{2}\left(\sqrt{s} - \sqrt{t}\right)^2.
\end{align*}
But this lower bound is actually attained by the coupling 
\begin{align*}
\gamma_\ast := \begin{bmatrix} 0 & 0 & s \\ 0 & 1-s-t & 0 \\ t & 0 & 0 \end{bmatrix},
\end{align*}
what completes the proof of (\ref{dist3}).

It remains to show that $W(\rho,\tau) > W(\rho,\sigma) + W(\sigma,\tau)$ for all $(s,t) \in \Delta^\circ$, i.e., that
\begin{align}
\label{dist4}
\sqrt{1+t-s - 2\sqrt{t(1-2s)}} > \sqrt{1 - 2\sqrt{s(1-s-t)}} + \sqrt{s}-\sqrt{t}
\end{align}
for all $0 < t < s < \tfrac{1}{2}$ and $2s+t < 1$. To this end, introduce $r := \tfrac{\sqrt{t}}{\sqrt{s}}$ and $a := \tfrac{\sqrt{1-2s} - \sqrt{t}}{\sqrt{s}}$ and notice that (\ref{dist4}) can be equivalently expressed as
\begin{align}
\label{dist5}
\sqrt{1+a^2} + r - 1 > \sqrt{\left(\sqrt{1+a^2+2ar}-1\right)^2+r^2}
\end{align}
for all $a>0$ and $0 < r < 1$. Since both sides of (\ref{dist5}) are positive we can square them, and after some cancellations obtain the following equivalent form of the desired inequality
\begin{align}
\label{dist6}
(1-r)\sqrt{1+a^2} + r(1+a) < \sqrt{1+a^2+2ar},
\end{align}
which is indeed true for all $a>0$ and $0 < r < 1$ by the strict concavity of the function $f(r) := \sqrt{1+a^2+2ar}$ on the interval $[0,1]$.


\section*{Acknowledgements}

The counterexamples were found with the assistance of \emph{ChatGPT 5.4 Pro}.


\begin{thebibliography}{1}

\bibitem{QOT} Sh. Friedland, M. Eckstein, S. Cole, K. \.{Z}yczkowski, \textit{Quantum Monge--Kantorovich problem and transport distance between density matrices}, Phys. Rev. Lett. \textbf{129}, 110402 (2022)
\bibitem{QOT2} S. Cole, M. Eckstein, Sh. Friedland, K. \.{Z}yczkowski, \textit{On quantum optimal transport}, Math. Phys. Anal. Geom. 26, 14 (2023)
\bibitem{Triangle} T. Miller, R. Bistro\'{n}, \textit{Distances between pure quantum states induced by a distance matrix}, \url{arXiv:2509.14727 [math-ph]} (2025)

\end{thebibliography}
\end{document}